%% file: Distances_Controllability.tex
\documentclass[10pt,twocolumn,twoside]{IEEEtran}
\usepackage{graphicx}
\usepackage{graphics} 
\usepackage{epsfig} 
\usepackage[abs]{overpic}
\usepackage{epstopdf}
\usepackage{amsmath}
\usepackage{amsfonts}
\usepackage{tikz}
\usepackage{color}
\usepackage{algorithm}
\usepackage{algpseudocode}
 
\usepackage{amsthm} 
\usepackage{todonotes}
\newtheorem{theorem}{\bf{Theorem}}[section]

\newtheorem{lem}[theorem]{Lemma}
\newtheorem{prop}[theorem]{Proposition}

\theoremstyle{plain}
\usepackage{soul}

\newcounter{defno}[section]
\renewcommand{\thedefno}{\arabic{section}.\arabic{defno}}

\newcounter{remno}[section]
\renewcommand{\theremno}{\arabic{section}.\arabic{remno}}

\newenvironment{definition}[1][Definition \thedefno]{\refstepcounter{defno}\begin{trivlist}
\item[\hskip \labelsep {\bfseries #1}]}{\end{trivlist}}

\newenvironment{remark}[1][Remark \theremno]{\refstepcounter{remno}\begin{trivlist}
\item[\hskip \labelsep {\bfseries #1}]}{\end{trivlist}}

\graphicspath{{figs/}}
\makeatletter
\def\input@path{{sections/}}
\makeatother

\newcommand*\circled[1]{\tikz[baseline=(char.base)]{
            \node[shape=circle,draw,inner sep=0.5pt] (char) {#1};}}

\title{\LARGE \bf
Graph Distances and Controllability of Networks
}
\author{A. Yasin Yaz{\i}c{\i}o\u{g}lu, Waseem Abbas, and Magnus Egerstedt\\
\thanks{A. Yasin Yaz{\i}c{\i}o\u{g}lu is with the Laboratory for Information and Decision Systems, Massachusetts Institute of Technology, {\tt yasiny@mit.edu}. \newline \indent
Waseem Abbas is with the Institute for Software Integrated Systems, Vanderbilt University, {\tt waseem.abbas@vanderbilt.edu}
\newline \indent Magnus Egerstedt is with the School of Electrical and Computer Engineering, Georgia Institute of Technology, {\tt magnus@gatech.edu}. \newline \indent The paper was presented in part at the $51^{\text{st}}$ IEEE Conference on Decision and Control, Maui, HI, December 10-13, 2012 (see \cite{Yasin12}). }
}

\IEEEoverridecommandlockouts

\begin{document}

\maketitle
\begin{abstract}
In this technical note, we study the controllability of diffusively coupled networks from a graph theoretic perspective. We consider leader-follower networks, where the external control inputs are injected to only some of the agents, namely the leaders. Our main result relates the controllability of such systems to the graph distances between the agents. More specifically, we present a graph topological lower bound on the rank of the controllability matrix. This lower bound is tight, and it is applicable to systems with arbitrary network topologies, coupling weights, and number of leaders. An algorithm for computing the lower bound is also provided. Furthermore, as a prominent application, we present how the proposed bound can be utilized to select a minimal set of leaders for achieving controllability, even when the coupling weights are unknown. 
\end{abstract}

\vspace{-0.15in}
\input{introduction}

\vspace{-0.15in}
\input{prelim}

\vspace{-0.15in}
\input{controllability}

\vspace{-0.15in}
\input{compute}

\vspace{-0.15in}
\input{lselect}

\vspace{-0.2in}
\input{conclusion}

\end{document}

%% file: introduction.tex
\section{Introduction}
Networks of diffusively coupled agents appear in numerous systems such as sensor networks (e.g., \cite{Speranzon06}), distributed robotics (e.g., \cite{Jadbabaie03}), power grids (e.g., \cite{Dorfler12}), social networks (e.g., \cite{Ghaderi13}), and biological systems (e.g., \cite{Gu15}). A central question regarding such networks is whether a desired global behavior can be induced by directly manipulating only a small subset of the agents, referred to as the ``leaders" in the network.  This question has motivated numerous studies on the controllability of networks. In particular, there has been a large interest in relating the \textit{network controllability} to the structure of the interaction graph. In this technical note, we relate the controllability of diffusively coupled agents with single integrator dynamics to the distances between the nodes on the interaction graph.

%

Various graph theoretic tools have recently been utilized to provide some topology based characterizations of network controllability. Some of the graph theoretic constructs that are widely employed for this purpose include \textit{equitable partitions} (e.g., \cite{Rahmani09,Egerstedt12}), \textit{maximum matching} (e.g., \cite{Liu11,Chapman13}), \textit{centrality based measures} (e.g., \cite{YLiu12,Pan14}), and \textit{dominating sets} (e.g., \cite{Nacher14}). Recently, the \textit{graph distances} have been used to acquire further insight on how the network structure and the locations of the leaders influence the network controllability \cite{Zhang11,Zhang14,Yasin12}. The graph distances, which rely purely on the shortest paths between the nodes, provide a computationally tractable and perceptible characterization of the graph structure. Thus, the distance-based relationships reveal some innate connections between the network topology and the network controllability.


The main contribution of this technical note is a distance-based tight lower bound on the dimension of the controllable subspace (Theroem \ref{lbound}). The bound is generic in the sense that it is applicable to systems with arbitrary network topologies, coupling weights, and number of leaders. Based on the distances between the leaders and the followers, we first define the distance-to-leaders (DL) vectors. Then, we define a certain ordering rule and derive the lower bound as the maximum length of the sequences of the DL vectors that satisfy the rule (Section \ref{controllability}). An algorithm to compute the proposed bound is also presented in Section \ref{compute}.  

A prominent attribute of the proposed bound is that, unlike the dimension of the controllable subspace, it does not depend on the coupling weights. Thus, the bound is useful in many applications, particularly when the information about the network is incomplete. As an example, in Section \ref{sec:lselect}, we present how the bound can be used in leader selection for achieving the controllability of any given network, even when the coupling weights are unknown. Finally, some conclusions are provided in Section \ref{conclusion}.

%% file: prelim.tex
\section{ Preliminaries}
\label{prelim}
\subsection{Graph Theory}
A graph $\mathcal{G}=(V,E)$ consists of a node set ${V=\{1,2,\hdots,n\}}$ and an edge set $E \subseteq V \times V$. For an undirected graph, each edge is represented as an unordered pair of nodes. For each $i \in V$, let $\mathcal{N}_i$ denote the \emph{neighborhood} of $i$, i.e., 
\begin{equation}
\label{neigh}
\mathcal{N}_i = \{ j \in V \mid (i,j) \in E \}
\end{equation}
A path between a pair of nodes $i,j \in V$ is a sequence of nodes $\{i, \hdots, j\}$ such that each pair of consecutive nodes are linked by an edge. The \emph{distance} between the nodes, $dist(i,j)$, is equal to the number of the edges that belong to the shortest path between the nodes. A graph is \emph{connected} if there exists a path between every pair of nodes. A graph is weighted if there is a corresponding weighting function $w: E \mapsto \mathbb{R}^+$. The \emph{adjacency matrix}, $\mathcal{A}$, of a weighted graph is defined as
\begin{equation}
\label{Adj}
\mathcal{A}_{ij}=\left\{\begin{array}{ll}w(i,j)&\mbox{ if } 
(i,j)\in E \\ 0&\mbox{ otherwise. }\end{array}\right.
\end{equation}
For any adjacency matrix $\mathcal{A}$, the corresponding \emph{degree matrix}, $\Delta$, is defined as
\begin{equation}
\label{Deg}
\Delta_{ij}=\left\{\begin{array}{ll}\sum_{k\in \mathcal{N}_{i}}\mathcal{A}_{ik}&\mbox{ if } 
i=j\\0&\mbox{ otherwise, }\end{array}\right.
\end{equation}
The \emph{graph Laplacian}, $L$, is defined as the difference of the degree matrix and the adjacency matrix, i.e.,
\begin{equation}
\label{Laplacian}
 L = \Delta-\mathcal{A}. 
\end{equation}

\vspace{-0.2in}
\subsection{Leader-Follower Networks}
A network of diffusively coupled agents can be represented as a graph, where the nodes correspond to the agents, and the weighted edges exist between the coupled agents. For such a network $\mathcal{G}=(V,E)$, let the dynamics of each agent $i\in V$ be
\begin{equation}
\label{consensus}
\dot{x}_i=\sum_{j\in \mathcal{N}_i}w(i,j)(x_j-x_i),
\end{equation}
where $x_i$ denotes the state of $i$,  and $w(i,j) \in \mathbb{R}^+$ represents the strength of the coupling between $i$ and $j$. 

In a \emph{leader-follower} setting, the objective is to drive the overall system by injecting external control inputs to some of the nodes, which are called the \emph{leaders}. The set of leaders can be represented as $\mathcal{L}=\{l_1, \hdots, l_m \} \subseteq V$, where, without loss of generality, the leaders are labeled such that $l_j<l_{j+1}$. For any leader-follower network, a global state vector $x$ can be obtained by stacking the states of all the nodes. Without loss of generality, let $x \in \mathbb{R}^n$, and let $u \in \mathbb{R}^m$ be the control input injected to the leaders. Then, the overall dynamics of the system can be expressed as 
\begin{equation}
\label{sys}
\dot{x}=-Lx+Bu,
\end{equation}
where  $B$ is an $n \times m$ matrix  with the following entries
\begin{equation}
\label{Beq}
B_{ij}=\left\{\begin{array}{ll}1&\mbox{ if $i=l_j$. } 
 \\ 0&\mbox{ otherwise. }\end{array}\right.
\end{equation}


For the system in (\ref{sys}), the \emph{controllable subspace} consists of the states that can be reached from $x(0)=\bold{0}$ in any finite time via an appropriate choice of $u(t)$. The controllable subspace is the range space of the \emph{controllability matrix}, i.e.,
\begin{equation}
\label{Gamma_M}
 \Gamma = \left[ \begin{array}{ccccc}
 B  & (-L)B  & (-L)^2B & \hdots & (-L)^{n-1}B \end{array}
\right] .
\end{equation}

%% file: controllability.tex
\section{Leader-Follower Distances and Controllability}
\label{controllability}

In this section, we present a connection between the controllability of networks and the distances of the nodes to the leaders on the interaction graph. More specifically, we utilize such distances to define a tight lower bound on the dimension of the controllable subspace, i.e. the rank of the controllability matrix. First, we provide some definitions prior to our analysis. 

\begin{definition} \emph{(Distance-to-Leaders (DL) Vector)}:
For each node $i$ in a network with $m$ leaders, the DL vector $d_i\in \mathbb{R}^m$ is defined as
\begin{equation}
\label{dldef}
d_{i,j} = dist(i,l_j),
\end{equation}
where $d_{i,j}$ denotes the $j^{th}$ entry of $d_i$, and $l_j$ denotes the $j^{th}$ leader for $j \in \{1,2, \hdots, m\}.$
\end{definition}
In our analysis, we utilize a specific sequence, which we define as a \emph{pseudo-monotonically increasing sequence}, of the DL vectors. For any vector sequence $D$, let $D_i$ be the $i^{th}$ vector in the sequence, and let $D_{i,j}$ denote the $j^{th}$ entry of $D_{i}$. 


\begin{definition}  \emph{(Pseudo-Monotonically Increasing (PMI) Sequence)}:
A sequence $D$ of vectors, where each vector is in $\mathbb{R}^m$, is a PMI sequence if for every $D_i$ there exists some $\alpha(i)\in\{1,2,\cdots,m\}$ such that

\begin{equation}
\label{rule}
D_{i,\alpha(i)} \; < \; D_{j,\alpha(i)},\;\;\forall i<j.
\end{equation}
\label{def:PMI}
\end{definition}
Condition (\ref{rule}) simply means that if $D_i$ is the $i^{th}$ vector in the sequence with $D_{i,\alpha(i)}$ being its $\alpha(i)^{th}$ entry, then the corresponding (i.e. $\alpha(i)^{th}$) entries of all the subsequent vectors in the sequence should be greater than $D_{i,\alpha(i)}$.

\textit{Example -- }
Consider the network shown in Fig. \ref{exfig}. For this network, one can build a PMI sequence of five DL vectors as
\begin{equation*}
\begin{split}
D & =    \left( D_1,  D_2,  D_3,  D_4,  D_5\right)  =
\left( d_1,  d_6, d_5,  d_3, d_4\right)\\
& = \left(\left[\begin{array}{cc}\circled{0}\\3\\\end{array}\right] ,  \left[\begin{array}{cc}3\\\circled{0}\\ \end{array}\right], \left[\begin{array}{cc}2\\\circled{1}\\ \end{array}\right], \left[\begin{array}{cc}\circled{1}\\2\\ \end{array}\right], \left[\begin{array}{cc} \circled{2}\\2\\\end{array}\right]\right),
\end{split}
\end{equation*}
where, for each $D_i$, the $\alpha(i)^{th}$ entry that satisfies (\ref{rule}) is encircled. For instance, consider $D_1$, for which $\alpha(1)=1$ and $D_{1,1}=0$. Note that the first entries of all the subsequent vectors are greater than 0. Similarly, for the second vector $D_2$, $\alpha(2)=2$ and $D_{2,2}=0$. Note that $D_{j,2}>0$ for $j>2$, and so on. For this example, another PMI sequence of five DL vectors could be build as $\left(d_1,  d_3,  d_6,  d_5,  d_4 \right)$, where $\alpha(1)= \alpha(2)=1$ and  $\alpha(3)= \alpha(4)= \alpha(5)= 2$. 

\begin{figure}[htb]
\begin{center}
\includegraphics[scale=0.8]{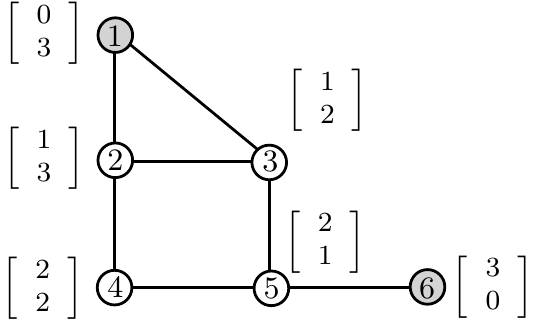}
\caption{A leader-follower network with two leaders (shown in gray), $l_1=1$ and $l_2=6$. Each node $i$ has its DL vector $d_i$ given next to itself. } 
\label{exfig}
\end{center}
\end{figure}


For any connected $\mathcal{G}=(V,E)$ and any set of leaders $\mathcal{L} \subseteq V$, let $\mathcal{D}_\mathcal{L}$ denote the set of all PMI sequences of the corresponding DL vectors. Furthermore, let $\delta_\mathcal{L}$ denote the length of the longest sequence in $\mathcal{D}_\mathcal{L}$, i.e. 
\begin{equation}
\label{lbeq}
\delta_\mathcal{L}=\underset{D \in \mathcal{D}_\mathcal{L}}{\textrm{max}}  |D|.
\end{equation}
In the following analysis, we show that, for any weighting function $w: E \mapsto \mathbb{R}^+$, the rank of the resulting controllability matrix is lower bounded by $\delta_\mathcal{L}$.

\begin{lem}
\label{lem1} Let $\mathcal{G}=(V,E)$ be a connected graph. Then, for any weighting function $w: E \mapsto \mathbb{R}^+$,
\begin{equation}
\label{propeq}
[(-L)^r ]_{ij} \left\{\begin{array}{ll} = 0&\mbox{ if $0\leq r < dist(i,j), $} 
 \\ \neq 0 &\mbox{ if $r = dist(i,j) $, }\end{array}\right.
\end{equation}
where $dist(i,j)$ is the distance between $i$ and $j$ on $\mathcal{G}$.
 \end{lem}
\begin{proof}
Using (\ref{Laplacian}), $(-L)^r$ can be expanded as
 \begin{equation}
\label{A_power}
(-L)^r=(\mathcal{A}-\Delta)^r=\mathcal{A}^r+\sum_{m=0}^{r-1}(-1)^{r-m}\mathcal{S}_{m},
\end{equation}
where $\mathcal{S}_{m}$ denotes the sum of all matrices that can be expressed as a multiplication of $m$ copies of $\mathcal{A}$ and $r-m$ copies of $\Delta$. For instance, if $r=2$, then $\mathcal{S}_0= \Delta^2$ and  $\mathcal{S}_1= \mathcal{A} \Delta + \Delta \mathcal{A}$.
Note that, for $w: E \mapsto \mathbb{R}^+$, any matrix that can be represented as such a multiplication has only non-negative entries since $\Delta$ and $\mathcal{A}$ have only non-negative entries. Moreover, $\Delta$ is a diagonal matrix and, for any connected graph, it has only positive entries on the main diagonal. As such, it does not alter the signs of entries when multiplied by a matrix. Hence, 

 \begin{equation}
\label{SA}
[\mathcal{S}_{m}]_{ij}=0 \Leftrightarrow [\mathcal{A}^{m}]_{ij}=0.
\end{equation}
Using (\ref{A_power}), 
\begin{equation}
\label{A_kb_i}
[(-L)^r]_{ij}=[\mathcal{A}^{r}]_{ij} +\sum_{m=0}^{r-1}(-1)^{r-m} [\mathcal{S}_{m}]_{ij}.
\end{equation}
Since $\mathcal{A}$ is the adjacency matrix of a weighted graph with positive edge weights, $[\mathcal{A}^{k}]_{ij}$ is equal to a positive scalar times the number of walks of length $k$ from node $i$ to node $j$. Hence, $[\mathcal{A}^{k}]_{ij}=0$ for all $0\leq k<dist(i,j)$, and $[\mathcal{A}^{dist(i,j)}]_{ij} \neq 0$ for any connected graph. Consequently, (\ref{SA}) and (\ref{A_kb_i}) together imply  (\ref{propeq}). 
 \end{proof}

\begin{theorem}
\label{lbound}
Let $\mathcal{G}=(V,E)$ be a connected graph, and let $\mathcal{L}\subseteq V$ be the set of leaders. Then, for any weighting function $w: E \mapsto \mathbb{R}^+$, the  controllability matrix $\Gamma$ satisfies   \begin{equation}
\label{lboundeq}
rank(\Gamma)\geq  \delta_\mathcal{L},
\end{equation}
where $\delta_{\mathcal{L}}$ is defined in (\ref{lbeq}).
\end{theorem}

\begin{proof}
For any connected $\mathcal{G}=(V,E)$ and any set of leaders $\mathcal{L}$, let $\{d_1,d_2,\cdots,d_{n}\}$ be the corresponding  DL vectors. Let $D=\left(D_1,D_2,\cdots,D_{\delta_{\mathcal{L}}}\right)$ be a PMI sequence of maximum length. Now, consider vectors of the form
\begin{equation}
\label{M_form}
 (-L)^{r_p}b_{\alpha(p)},
\end{equation} 
where $\alpha(p)$ is the index of $D_p$ as per the definition of a PMI sequence (Definition \ref{def:PMI}), $r_p = D_{p,\alpha(p)}$, and $b_{\alpha(p)}$ denotes the $\alpha(p)^{th}$ column of the input matrix $B$. If $D_p$ corresponds to the DL vector of node $i$, i.e. $D_p = d_i$, then  $r_p = d_{i,\alpha(p)}$. As a result of Lemma \ref{lem1}, $i^{th}$ entry of the vector in (\ref{M_form}) is non-zero. Also, for any node $j$ with $d_{j,\alpha(i)} > d_{i,\alpha(i)}$, the $j^{th}$ entry of the vector in (\ref{M_form}) equal to zero. Using this along with the definition of PMI sequences, we conclude that the $n \times \delta_\mathcal{L}$ matrix \begin{equation}
\label{full_rank}
\left[ \begin{array}{ccccc}
  (-L)^{r_1}b_{\alpha(1)} &   (-L)^{r_2}b_{\alpha(2)}  &\hdots & (-L)^{r_{ \delta_{\mathcal{L}}}}b_{\alpha( \delta_{\mathcal{L}})} \end{array}
\right],   
\end{equation}
has a full column rank since each column contains the left-most non-zero entry in some rows. Note that for every $p \in \{1,2,\hdots, \delta_{\mathcal{L}}\}$, we have $r_p=D_{p,\alpha(p)} \leq n-1$ since the distance between any two nodes is always smaller than or equal to $n-1$. Hence, each column of the matrix in (\ref{full_rank}) is also a column of $\Gamma$. Consequently, $rank(\Gamma) \geq \delta_\mathcal{L}$. 
\end{proof}

Since (\ref{lboundeq}) holds for any weighting function ${w: E \mapsto \mathbb{R}^+}$, the proposed lower bound is closely related to the notion of  \textit{strong structural controllability}.  A network is said to be structurally controllable if and only if there exists ${w: E \mapsto \mathbb{R}^+}$ such that the resulting controllability matrix is of full rank \cite{Lin74}. Furthermore, a network is said to be strongly structurally controllable if and only if for any ${w: E \mapsto \mathbb{R}^+}$, the resulting controllability matrix is of full rank \cite{Mayeda79}. In this regard, the dimension of controllable subspace in the sense of strong structural controllability can be defined as the minimum of the ranks of all controllability matrices obtained by using arbitrary weighting functions ${w: E \mapsto \mathbb{R}^+}$ \cite{Jarczyk13}. Accordingly, $\delta_{\mathcal{L}}$ is essentially a lower bound on the dimension of controllable subspace in the sense of strong structural controllability. 

We also emphasize that the lower bound in Theorem \ref{lbound} is tight, i.e., there exist $\mathcal{G} =(V,E)$, $\mathcal{L} \subseteq V$, and ${w: E \mapsto \mathbb{R}^+}$ such that $rank(\Gamma)=\delta_{\mathcal{L}}$. Any cycle graph with two adjacent nodes being the leaders, or any path graph with a leaf node being the leader are some of the examples that satisfy (\ref{lboundeq}) with equality. In these two examples, $rank(\Gamma)=\delta_\mathcal{L}$ follows from the fact that both cases lead to $\delta_\mathcal{L}=n$, where $n$ is the total number of nodes in the graph. In the remainder of this section, we present the connections between the proposed lower bound and some closely related distance-based measures, namely the maximum distance to the leaders and the number of unique DL vectors. 

For networks with a single leader, any PMI sequence of the DL vectors consists of one dimensional vectors with monotonically increasing entries. Hence, for such networks, $\delta_{\mathcal{L}}$ is equal to one plus the maximum distance to the leader, which was indeed proposed in \cite{Zhang11} as a lower bound of the controllability matrix for single-leader networks. An extension to the case of multiple leaders was later presented in \cite{Zhang14} by taking the maximum of this value among all the leaders, i.e.

\begin{equation}
\label{dmax}
\mu_\mathcal{L}=\underset{i \in V, j \in \mathcal{L}}{\textrm{max}}dist(i,j)+1.
\end{equation}

The relationship between $\mu_\mathcal{L}$ and  $\delta_\mathcal{L}$ can be seen through the following observation: If one considers only the PMI sequences that satisfy (\ref{rule}) for some fixed entry, i.e.  $\alpha(1)= \hdots = \alpha(|D|)$, then the longest PMI sequence in this constrained set, $\mathcal{D}_\mathcal{L}^* \subseteq \mathcal{D}_\mathcal{L}$, has length $\mu_\mathcal{L}$. Hence, one can conclude that the following inequality holds for any leader-follower network:
 \begin{equation}
\label{dmaxcomp}
rank(\Gamma)\geq  \delta_\mathcal{L} \geq \mu_\mathcal{L}.
\end{equation}
In light of (\ref{dmaxcomp}), while the two quantities are equal for single-leader networks, the proposed lower bound $\delta_\mathcal{L}$ is richer than $\mu_\mathcal{L}$ in capturing the controllability of networks with multiple leaders. In fact, since the maximum distance between any two nodes in a graph is the diameter of the graph by definition, $\mu_\mathcal{L}$ is always less than or equal to one plus the diameter of the graph, even when every node is a leader. In general, the difference between $\delta_\mathcal{L}$ and $\mu_\mathcal{L}$ depends on the graph topology and the leader assignment. For instance, the example in Fig.~\ref{exfig} yields $\delta_\mathcal{L}=5$ and $\mu_\mathcal{L}=3$. Numerical comparisons of the bounds for Erd\"{o}s-Renyi and Barab\'{a}si-Albert graphs with two leaders are illustrated in Fig. \ref{fig:comparison}. In this figure, each point on the plot corresponds to the average value for 50 randomly generated cases (each case is a randomly generated graph and a pair of randomly assigned leaders). The results indicate that $\delta_\mathcal{L}$ provides a significantly better utilization of the graph distances in the controllability analysis of multi-leader networks, even when the network has only a pair of leaders. 



\begin{figure}[htb]
\begin{center}
\includegraphics[scale=0.27]{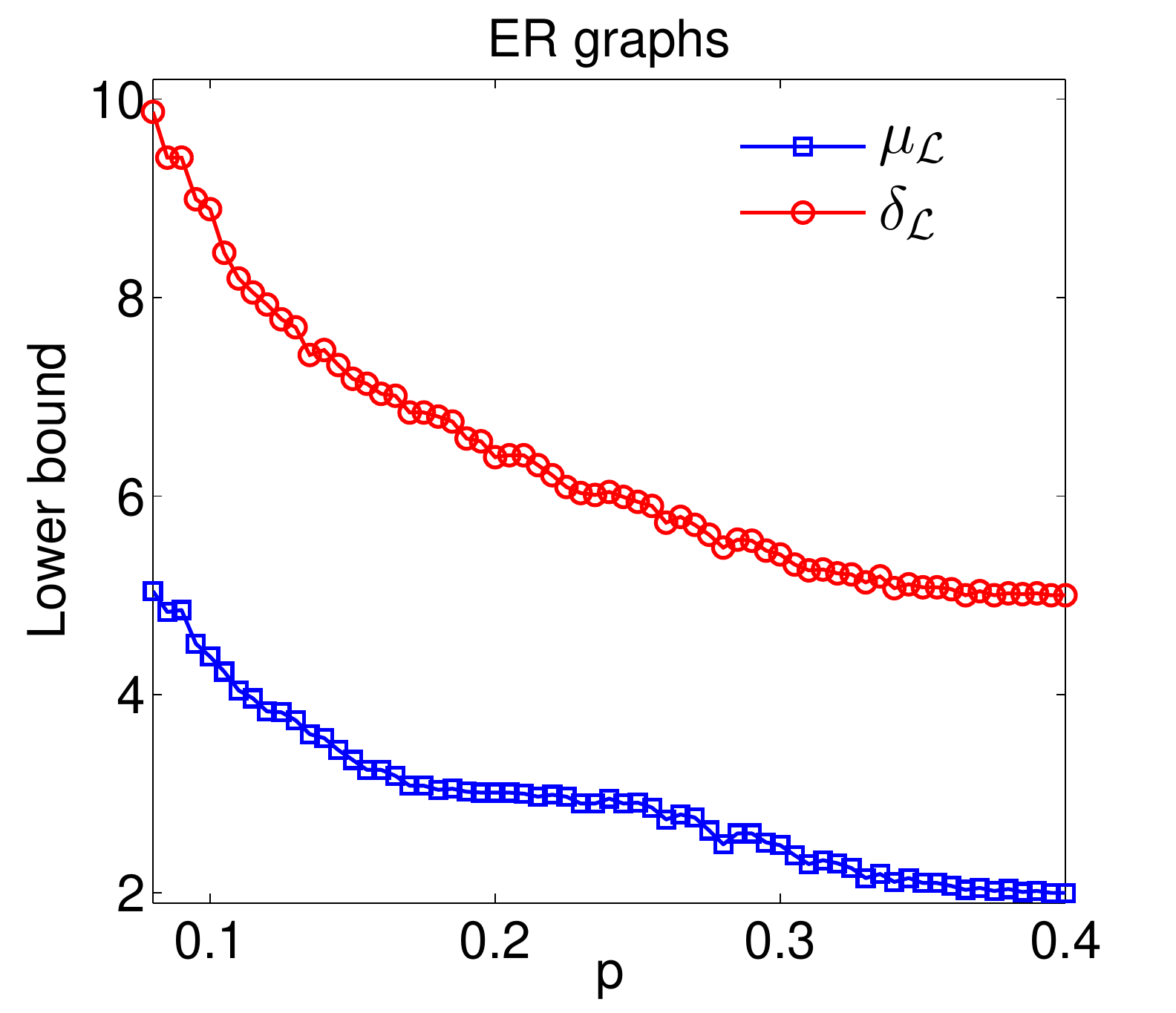}
\includegraphics[scale=0.27]{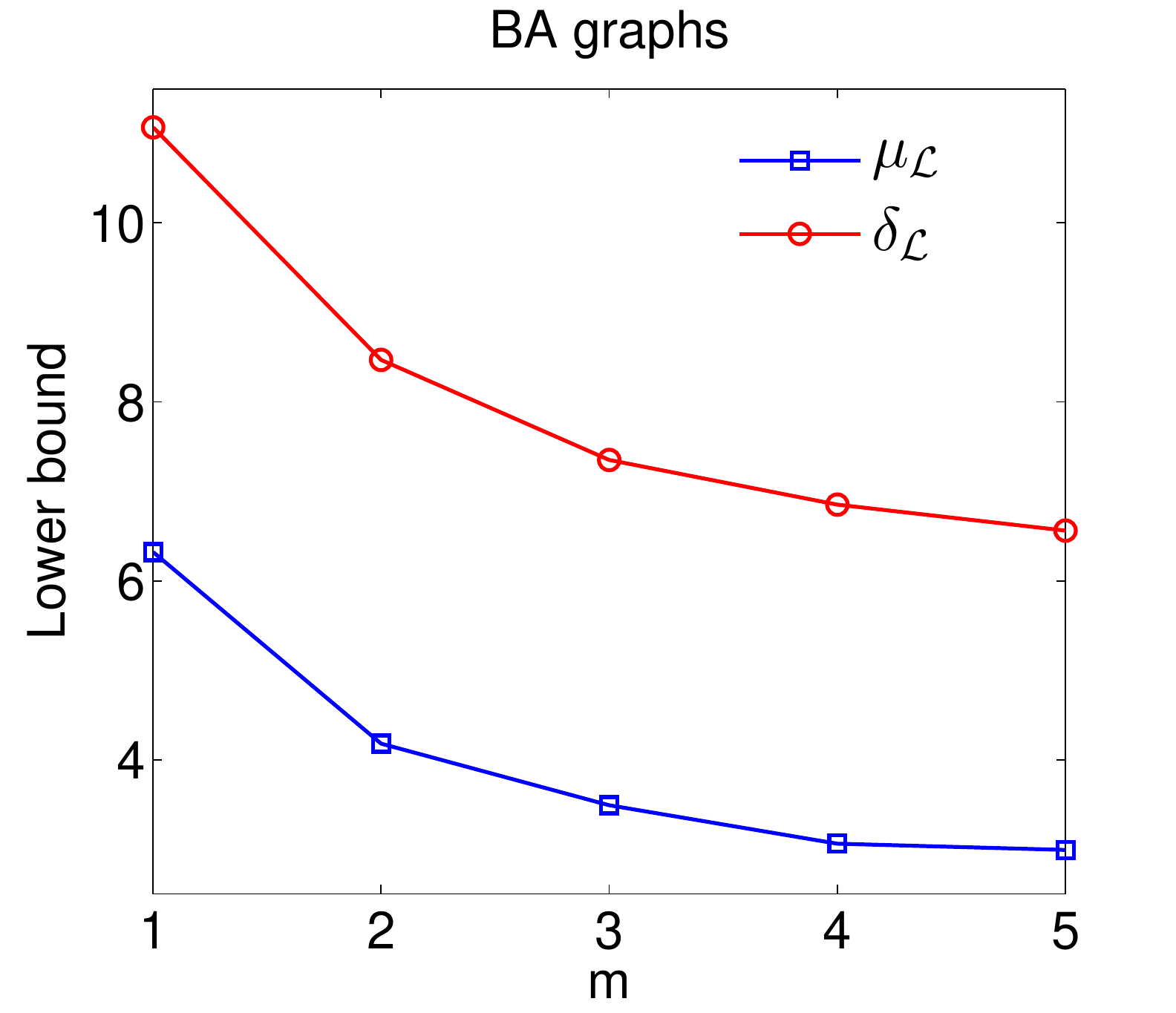}
\caption{Comparison of the lower bounds for two randomly selected leaders on Erd\"{o}s-Renyi random graphs with 50 nodes in which any two nodes are adjacent with the probability $p$; and Barab\'{a}si-Albert graph with 50 nodes in which each new node is connected to $m$ existing nodes through a preferential attachment strategy.}
\label{fig:comparison}
\end{center}
\end{figure}

The proposed lower bound is also closely related to the number of unique DL vectors. For any connected $\mathcal{G}=(V,E)$ and any leader set $\mathcal{L}\subseteq V$, let $\upsilon_\mathcal{L}$ be the number of unique DL vectors. Note that, due to (\ref{rule}), each vector in a PMI sequence has an entry that is strictly smaller than the corresponding entries of all the following vectors in that sequence. Hence, a PMI sequence cannot contain two identical vectors. Consequently, the proposed lower bound is always less than or equal to the number of unique DL vectors, i.e. 
 \begin{equation}
\label{udlb}
\upsilon_\mathcal{L} \geq \delta_\mathcal{L}.
\end{equation}

The relationship in (\ref{udlb}) facilitates a deeper understanding for potential applications of the proposed lower bound. For instance, for any $k \in \{1, \hdots,n\}$, having $\upsilon_\mathcal{L}=k$ is a necessary condition to have $\delta_\mathcal{L}=k$. Hence, it is possible to conclude that a network is completely controllable as per the proposed lower bound, only if each follower has a distinctive DL vector. The relationships in (\ref{dmaxcomp}) and (\ref{udlb}) may naturally yield the question of whether $\upsilon_\mathcal{L}$ can capture $rank(\Gamma)$ better than $\delta_\mathcal{L}$ and $\mu_\mathcal{L}$. In~the following result, we show that there is no such universal relationship between $rank(\Gamma)$ and $\upsilon_\mathcal{L}$. 



%

\begin{prop}
\label{udist2}
For leader-follower networks, the number of unique DL vectors, $\upsilon_\mathcal{L}$, is not a universal bound of $rank(\Gamma)$. \end{prop}
\begin{proof}
We prove this statement by providing some examples both for $rank(\Gamma)>\upsilon_\mathcal{L}$ and for $rank(\Gamma)<\upsilon_\mathcal{L}$. 

1) A path graph with uniform edge weights is controllable from any single node if and only if the the number of nodes is a power of two, i.e. $n=2^k$ for some $k \in \mathbb{N}$ \cite{Parlangeli12}.  Note that, for a path graph with a single leader, $\delta_\mathcal{L}=n$ if and only if the leader is a leaf node. Hence, for any $n=2^k$, if a non-leaf node is the only leader of a path graph with uniform edge weights, then $rank(\Gamma)> \upsilon_\mathcal{L}$.

2) A cycle graph with uniform edge weights is controllable from any pair of nodes if and only if the the number of nodes is a prime number \cite{Parlangeli12}. Consider any cycle graph of $n$ nodes such that $n$ is an odd composite number. For such a graph with two leaders, the clockwise and counterclockwise paths between the leaders have different lengths. Hence, any pair of nodes that have equal distances from one of the leaders have different distances from the other leader, i.e.  $\upsilon_\mathcal{L}=n$. Furthermore, in light of \cite{Parlangeli12}, there exists a pair of nodes that render the system uncontrollable if they are assigned as the leaders. If such a pair is assigned as the leaders, then $rank(\Gamma)< \upsilon_\mathcal{L}$.
\end{proof}

We would like to conclude this section with a remark regarding the application of the presented results to directed networks.

\begin{remark} (\emph{Directed Networks}) \label{directrem}The formulation in (\ref{neigh})-(\ref{Gamma_M}) is applicable to directed networks with the interaction graph $\mathcal{G}=(V,E)$, where each $(i,j) \in E$ denotes that $i$ is influenced by $j$ as in (\ref{consensus}). For such a network, powers of the adjacency matrix $\mathcal{A}$ have the property that $[\mathcal{A}^r]_{ij}=0$ for all $0\leq r < dist(i,j)$ and $[\mathcal{A}^{dist(i,j)}]_{ij}\neq0$, where $dist(i,j)$ is the length of the shortest directed path from $i$ to $j$. Hence, by using the corresponding DL vectors as in (\ref{dldef}), the results in both Lemma 3.1 and Theorem 3.2 can be extended to strongly connected networks (i.e., there exists a directed path from every node to every other node). 
\end{remark}

%% file: compute.tex
\section{An Algorithm for Computing the lower bound}
\label{compute}

In this section, we present an algorithm to compute the proposed lower bound, $\delta_\mathcal{L}$. Note that the main contribution of this work is the lower bound itself, and the algorithm in this section is provided to facilitate some practical use of our result.  

Let $S=\{d_1, d_2,\hdots, d_n \}$ be the set of all DL vectors for a given leader-follower network. Given these vectors, we present an iterative way of generating the longest PMI sequences. Let $C_p$ be the set of all DL vectors that can be assigned as the $p^{th}$ element of such a sequence $D$. According to these definitions, $C_1=S$. Once a vector from $C_p$ is assigned as the $p^{th}$ element of the sequence, $D_p$, and an index $\alpha(p)$ satisfying (\ref{rule}) is chosen, the resulting $C_{p+1}$ can be obtained from $C_p$ as

\begin{equation}
\label{C_set}
 {C}_{p+1}=  \{d_i \in {C}_p \mid d_{i,\alpha(p)} > D_{p,\alpha(p)}\}.
  \end{equation}

In order to obtain longer sequences, this iteration must be continued until $C_p= \emptyset$.  However, in general there are too many possible sequences that can be obtained this way, and it is not feasible to find the maximum length for PMI sequences by searching among all these possibilities. Instead, we present a necessary condition for a PMI sequence to have the maximum possible length. This necessary condition significantly lowers the number of sequences to consider.
\begin{lem}
 \label{algorithm}
Let $D$ be a PMI sequence of DL vectors with the maximum possible length, then its $p^{th}$ entry, $D_p$, satisfies 
\begin{equation}
\label{seq_eq}
D_{p,\alpha(p)}= \min_{d_i\in C_p}d_{i,\alpha(p)}.
\end{equation}
\end{lem}
\begin{proof}
For the sake of contradiction, assume that this is not true for a PMI sequence $D$ with the maximum length. Then, there exists a DL vector $d_j \in  C_p$ such that $d_{j,\alpha(p)} < D_{p,\alpha(p)}$. By the construction of a PMI sequence, $d_j$ can not be added to this sequence after $D_p$.  However, $d_j$ can be added right before $D_p$ while keeping all the other parts of $D$ the same since $d_{j,\alpha(p)}$ can be selected to satisfy (\ref{rule}) in the resulting sequence. Hence, it is possible to obtain a longer PMI sequence, which leads to the contradiction that $D$ does not have the maximum possible length.
\end{proof}

In light of (\ref{C_set}), as far as the sequence length is concerned, the only important decision at each step $p$ in building a PMI sequence satisfying (\ref{seq_eq}) is the choice of $\alpha(p)$. Based on this observation, we propose Algorithm \ref{lbalg} for computing the lower bound.

\begin{algorithm}
\caption{}
\label{lbalg}
\begin{algorithmic}[1]
\footnotesize
\State \textbf{initialize:} $\mathcal{C}_1=\{\{d_1, d_2, \hdots, d_n\}\}$; $p=1$
\While{$\mathcal{C}_{p,y}\ne\emptyset$ for some $y\in\{1,\cdots,|\mathcal{C}_p|\}$}
\State $q=1$
\For{$i=1:|\mathcal{C}_p|$}
\If{$\mathcal{C}_{p,i}\neq \emptyset$}
\For{$j=1:m$}
\State $\mathcal{C}_{p+1,q} =  \{d_t\in\mathcal{C}_{p,i}\mid\;d_{t,j} > \min\limits_{d_s\in\mathcal{C}_{p,i}}d_{s,j}\}$
\State $q=q+1$
\EndFor
\EndIf
\EndFor
\State $p=p+1$
\EndWhile
\State \textbf{return} $p-1$
\end{algorithmic}
\end{algorithm}
\normalsize

In Algorithm \ref{lbalg}, the variable $\mathcal{C}_p$ is the multiset, where each element $\mathcal{C}_{p,i}$ is the $C_p$ resulting from (\ref{C_set}) for specific choices of $\alpha(1), \hdots, \alpha(p-1)$, subject to the corresponding PMI sequences satisfying (\ref{seq_eq}). The main while loop iterates as long as there exists a longer PMI sequence that satisfies (\ref{seq_eq}). Note that for each $\mathcal{C}_{p,i} \neq \emptyset$, there are $m$ (number of leaders) different $\mathcal{C}_{p+1,q}$, each corresponding to a particular choice of $\alpha(p)$. As such, Algorithm \ref{lbalg} computes $\delta_\mathcal{L}$ by generating no more than $m^{\delta_\mathcal{L}}$ elements $\mathcal{C}_{p,i}$. 

\begin{prop}
 \label{algorithm_prop}
Given the DL vectors for any connected leader-follower network, Algorithm \ref{lbalg} returns $\delta_\mathcal{L}$.
\end{prop}
\begin{proof}
By combining (\ref{C_set}) and (\ref{seq_eq}) (line 7 of Algorithm \ref{lbalg}), at each step $p$, Algorithm \ref{lbalg} builds all the possible $C_{p+1}$ that correspond to the PMI sequences of $p$ vectors satisfying (\ref{seq_eq}). Hence, when $\mathcal{C}_{p,y} = \emptyset$ for every $y\in\{1,\cdots,|\mathcal{C}_p|\}$, there is not a longer PMI sequence that satisfies the necessary condition in Lemma \ref{algorithm}. Consequently, Algorithm \ref{lbalg} always returns the length of the longest PMI sequence, $\delta_\mathcal{L}$. 
\end{proof}

Note that, in light of Remark \ref{directrem}, Algorithm \ref{lbalg} can also be used to compute the proposed lower bound for any directed network with a strongly connected interaction graph.

For a sample run of Algorithm \ref{lbalg}, consider the network in Fig. \ref{exfig}. For this example, the algorithm terminates after the fifth iteration of the while loop, and $\delta_\mathcal{L}=5$. The resulting flow of Algorithm \ref{lbalg} can be represented via a tree diagram as illustrated in Fig. \ref{fig:graph}. 

\begin{figure}[htb]
\begin{center}
\includegraphics[scale=0.65]{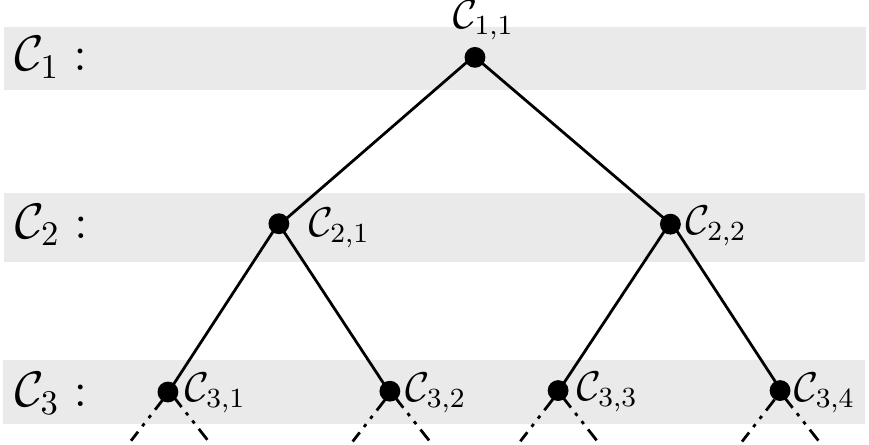}
\caption{An illustration of the flow of Algorithm \ref{lbalg} for the network in Fig. \ref{exfig}.} 
\label{fig:graph}
\end{center}
\end{figure}

In Fig. \ref{fig:graph}., each node $\mathcal{C}_{p,i}$ at a given level $p>1$ corresponds to an element of $\mathcal{C}_p$ that is computed in the $(p-1)^{st}$ iteration of the while loop. The left child of a node $\mathcal{C}_{p,i}$ is obtained from $\mathcal{C}_{p,i}$ by deleting all the DL vectors whose first entries are equal to the minimum value of the first entries among all the DL vectors in $\mathcal{C}_{p,i}$, i.e., obtained from $\mathcal{C}_{p,i}$ as per (\ref{seq_eq}) and (\ref{C_set}) for $\alpha(p)=1$. Similarly, the right child of a node $\mathcal{C}_{p,i}$ is obtained by following the same procedure for $\alpha(p)=2$. Accordingly, the elements in the first three levels are  
{\fontsize{0.22cm}{3pt}
$$
\begin{array}{llll}
\mathcal{C}_{1,1}= \left\{
\left[\begin{array}{c} 0\\3 \end{array} \right],
\left[\begin{array}{c} 1\\2 \end{array} \right],
\left[\begin{array}{c} 1\\3 \end{array} \right],
\left[\begin{array}{c} 2\\1 \end{array} \right],
\left[\begin{array}{c} 2\\2 \end{array} \right],
\left[\begin{array}{c} 3\\0 \end{array} \right]
\right\};\\
\mathcal{C}_{2,1}=\left\{
\left[\begin{array}{c} 1\\2 \end{array} \right],
\left[\begin{array}{c} 1\\3 \end{array} \right],
\left[\begin{array}{c} 2\\1 \end{array} \right],
\left[\begin{array}{c} 2\\2 \end{array} \right],
\left[\begin{array}{c} 3\\0 \end{array} \right]
\right\};\\
\mathcal{C}_{2,2} = \left\{
\left[\begin{array}{c} 0\\3 \end{array} \right],
\left[\begin{array}{c} 1\\2 \end{array} \right],
\left[\begin{array}{c} 1\\3 \end{array} \right],
\left[\begin{array}{c} 2\\1 \end{array} \right],
\left[\begin{array}{c} 2\\2 \end{array} \right]
\right\};\\
\mathcal{C}_{3,1}=\left\{
\left[\begin{array}{c} 2\\1 \end{array} \right],
\left[\begin{array}{c} 2\\2 \end{array} \right],
\left[\begin{array}{c} 3\\0 \end{array} \right]
\right\}; \\

\mathcal{C}_{3,2} = \mathcal{C}_{3,3}=\left\{
\left[\begin{array}{c} 1\\2 \end{array} \right],
\left[\begin{array}{c} 1\\3 \end{array} \right],
\left[\begin{array}{c} 2\\1 \end{array} \right],
\left[\begin{array}{c} 2\\2 \end{array} \right]
\right\};\\
\mathcal{C}_{3,4}=\left\{
\left[\begin{array}{c} 0\\3 \end{array} \right],
\left[\begin{array}{c} 1\\2 \end{array} \right],
\left[\begin{array}{c} 1\\3 \end{array} \right],
\left[\begin{array}{c} 2\\2 \end{array} \right]
\right\}.
\end{array}
$$
}

%% file: lselect.tex
\section{Distance-Based Leader Selection for Controllability}
\label{sec:lselect}
One of the main attributes of the proposed lower bound $\delta_\mathcal{L}$ is its independence of the edge weights. In contrast, $rank(\Gamma)$ depends on the values of the edge weights unless there is some constraint such as all the weights being equal. Consequently, computing $\delta_\mathcal{L}$ typically requires significantly less information about the overall system. This minimality of the required information makes the lower bound attractive in many applications such as leader selection for controllability.  

Leader selection problems typically require finding a leader set $\mathcal{L}$ that optimizes a system objective such as robustness, mixing time, or controllability (e.g., \cite{Lin14,Clark12,Yasin13,Aguilar15}).  For instance, consider the problem of finding a minimum number of leaders that render a given network controllable under the resulting leader-follower dynamics. If the edge weights are known (or if they are known to be identical), then the rank of the controllability matrix can be computed for any set of leaders. Hence, a possible, yet not scalable, way to find a minimal set of leaders for controllability is to execute an exhaustive search. Note that, aside from the complexity issues, the rank computation is not applicable if the edge weights are unknown and arbitrary. In such cases, the leader selection problem needs to be solved by leveraging the structural properties of the interaction graph.

One approach to achieving controllability under arbitrary coupling weights is to choose the minimal $\mathcal{L}$ that achieves structural controllability (e.g., \cite{Liu11}). Structural controllability implies that the selected leaders provide complete controllability for some, not all, weighting functions ${w: E \mapsto \mathbb{R}^+}$. Hence, this approach may fall short in some applications, especially when there are constraints on the admissible edge weights. For instance, if all the edge weights in a network are equal by design, then it is known that a complete graph is not controllable by any single leader \cite{Rahmani09} whereas a single leader is enough to achieve structural controllability \cite{Liu11}.  Alternatively, the notion of strong structural controllability can be employed in the leader selection (e.g., \cite{Chapman13}). In this regard, the proposed bound  $\delta_\mathcal{L}$ can be used to ensure that the dimension of the controllable subspace is not smaller than some desired value, $k \in \{1,2, \hdots, n\}$, for any $w: E \mapsto \mathbb{R}^+$ by formulating the leader selection problem as

 \begin{equation}
 \label{lselect}
 \begin{aligned}
 & \underset{\mathcal{L} \in 2^V}{\text{minimize}} 
 & & |\mathcal{L}| \\
 & \text{subject to}
 & & \delta_\mathcal{L} \geq k,  \\
 \end{aligned}
\end{equation}

In light of Theorem \ref{lbound}, any element in the feasible set of the problem in (\ref{lselect}) renders $rank(\Gamma) \geq k$ for any weighting function $w: E \mapsto \mathbb{R}^+$. Note that the problem in (\ref{lselect}) is always feasible, i.e.,  for any given $\mathcal{G}=(V,E)$ and $k \in \{1,2, \hdots, n\}$, there always exists some ${\mathcal{L} \subseteq V}$ such that $\delta_\mathcal{L}\geq k$. Indeed, the feasibility can be shown by considering the trivial case, ${\mathcal{L}=V}$. In that case, for each DL vector, the entry that contains the distance of the corresponding node from itself satisfies (\ref{rule}) for any sequence of the corresponding DL vectors. Consequently, any sequence of the DL vectors is a PMI sequence if $\mathcal{L}=V$, and $\delta_V=n$. Note that the feasibility would not be guaranteed if the problem was posed by using $\mu_{\mathcal{L}}$ instead of $\delta_{\mathcal{L}}$ since $\mu_{\mathcal{L}}$ is always upper bounded by one plus the graph diameter. An example, where the leaders were assigned by solving (\ref{lselect}) for $k=n$, is illustrated in Fig. \ref{ex_ls}. For this example, the leaders were selected through an exhaustive search by first looking for a single-leader solution and incrementing the number of leaders until a solution exists. Algorithm \ref{lbalg} was used to compute $\delta_\mathcal{L}$ for each candidate $\mathcal{L}$.
\begin{figure}[htb]
\begin{center}
\includegraphics[scale=0.75]{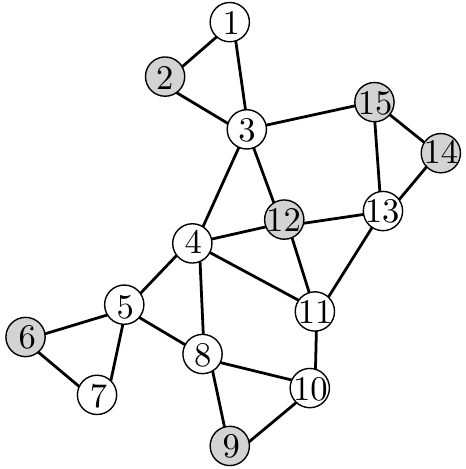}
\caption{ A graph of 15 nodes, $\mathcal{G}=(V,E)$, and a minimal selection of leaders (shown in gray), $\mathcal{L}=\{2,6,9,12,14,15\}$,  such that $\delta_\mathcal{L}=15$. The network is completely controllable via $\mathcal{L}$ for any weighting function $w: E \mapsto \mathbb{R}^+$.}
\label{ex_ls}
\end{center}
\end{figure}


%% file: conclusion.tex
\section{Conclusion}
\label{conclusion}
In this technical note, we presented that the distances between the leaders and the followers on the interaction graph contain some fundamental information about the controllability of the leader-follower networks. In particular, we used the distance-to-leaders (DL) vectors to derive a tight lower bound on the dimension of the controllable subspace. The proposed bound is applicable to networks with arbitrary interaction graphs and  weighting functions $w: E \mapsto \mathbb{R}^+$. We also provided some connections between the proposed lower bound and a pair of closely related distance-based measures, namely the maximum distance from the leaders and the number of distinct DL vectors. While the results were presented for undirected networks, we also showed how they can be extended to directed networks. Furthermore, we presented an algorithm for computing the lower bound. The proposed bound may find its applications in various networked control problems, especially when the edge weights are unknown. As a prominent application, we presented how it can be utilized to find a minimal set of leaders that ensure the controllability of a leader-follower network with a given interaction graph under any weighting function $w: E \mapsto \mathbb{R}^+$. 
